\documentclass[11pt]{article} 

\usepackage[protrusion=true,expansion=true]{microtype} 
\usepackage{graphicx} 
\usepackage{wrapfig} 

\usepackage{mathpazo} 
\usepackage[T1]{fontenc} 
\usepackage{amsmath,amsfonts,amssymb,amsthm,commath,dsfont,comment}
\usepackage{algorithmicx}
\usepackage{algpseudocode}
\usepackage{algorithm}
\usepackage[page]{appendix}
\usepackage{array,multirow}
\usepackage{color}
\usepackage{fullpage}

\newcommand{\Sep}{\textit{Sep}}
\newcommand{\Res}{\textit{Res}}
\newcommand{\Em}{\textit{Em}}
\newcommand{\Pa}{\textit{Pa}}
\newcommand{\Ch}{\textit{Ch}}

\newcommand{\Ex}{\textit{Explore}}
\newcommand{\Size}{\textit{Size}}
\newcommand{\MEC}{\textit{MEC}}

\newtheorem{lemma}{Lemma}
\newtheorem{claim}{Claim}
\newtheorem{proposition}{Proposition}
\newtheorem{theorem}{Theorem}
\newtheorem{corollary}{Corollary}

\newtheorem{definition}{Definition}
\newtheorem{remark}{Remark}
\newtheorem{example}{Example}

\makeatletter
\renewcommand\@biblabel[1]{\textbf{#1.}} 
\renewcommand{\@listI}{\itemsep=0pt} 

\renewcommand{\maketitle}{ 
\begin{center}
{\LARGE\@title} 
\end{center}

\vspace{50pt} 

{\large\@author} 

\vspace{40pt} 
}

\title{\textbf{Counting and Sampling from Markov Equivalent DAGs}\\\vspace{5mm}
 \textbf{Using Clique Trees }}

\author{\noindent
\textsc{AmirEmad Ghassami}\\
{\textit{Department of ECE, University of Illinois at Urbana-Champaignn, Urbana, IL, USA}}\\\vspace{3mm}
{\texttt{ghassam2@illinois.edu}}\\
\textsc{Saber Salehkaleybar}\\
{\textit{Electrical Engineering Department, Sharif University of Technology, Tehran, Iran}}\\\vspace{3mm}
{\texttt{saleh@sharif.edu}}\\
\textsc{Negar Kiyavash}\\
{\textit{School of ISyE and ECE, Georgia Institute of Technology, Atlanta, GA, USA}}\\\vspace{3mm}
{\texttt{nkiyavash3@gatech.edu}}\\
\textsc{Kun Zhang}\\
{\textit{Department of Philosophy, Carnegie Mellon University, Pittsburgh, USA}}\\\vspace{3mm}
{\texttt{kunz1@cmu.edu}}\\
}

\begin{document}

\maketitle 



\begin{abstract}
A directed acyclic graph (DAG) is the most common graphical model for representing causal relationships among a set of variables. When restricted to using only observational data, the structure of the ground truth DAG is identifiable only up to Markov equivalence, based on conditional independence relations among the variables. Therefore, the number of DAGs equivalent to the ground truth DAG is an indicator of the causal complexity of the underlying structure--roughly speaking, it shows how many interventions or how much additional information is further needed to recover the underlying DAG. In this paper, we propose a new technique for counting the number of DAGs in a Markov equivalence class. Our approach is based on the clique tree representation of chordal graphs. We show that in the case of bounded degree graphs, the proposed algorithm is polynomial time. We further demonstrate that this technique can be utilized for uniform sampling from a Markov equivalence class, which provides a stochastic way to enumerate DAGs in the equivalence class and may be needed for finding the best DAG or for causal inference given the equivalence class as input. We also extend our counting and sampling method to the case where prior knowledge about the underlying DAG is available, and present applications of this extension in causal experiment design and estimating the causal effect of joint interventions.
\end{abstract}


\vspace{30pt} 


\section{Introduction}
\label{sec:intro}

A directed acyclic graph (DAG) is a commonly used graphical model to represent causal relationships among a set of variables \cite{pearl2009causality}. In a DAG representation, a directed edge $X\rightarrow Y$ indicates that variable $X$ is a direct cause of variable $Y$ relative to the considered variable set. Such a representation has numerous applications in fields ranging from biology \cite{sachs2005causal} and genetics \cite{zhang2013integrated} to machine learning \cite{peters2017elements,koller2009probabilistic}.

The general approach to learning a causal structure is to use statistical data from variables to find a DAG, which is maximally consistent with the conditional independencies estimated from data. This is due to the fact that under Markov property and faithfulness assumption,  d-separation of variables in a DAG is equivalent to conditional independencies of the variables in the underlying joint probability distribution \cite{spirtes2000causation,pearl2009causality}. However, a DAG representation of a set of conditional independencies in most cases is not unique. This restricts the learning of the  structure to the Markov equivalences of the underlying DAG. The set of Markov equivalent DAGs is referred to as a Markov equivalence class (MEC). A MEC is commonly represented by a mixed graph called essential graph, which can contain both directed and undirected edges \cite{spirtes2000causation}.

In general, there is no preference amongst the elements of MEC, as they all represent the same set of conditional independencies. Therefore, for a given dataset (or a  joint probability distribution), the size of the MEC, i.e., the number of its elements, can be seen as a metric for measuring the causal complexity of the underlying structure--this complexity indicates how many interventional experiments or how much additional information (e.g., knowledge about the causal mechanisms) is needed to further fully learn the DAG structure, which was only partially recovered from mere observation.
Another important application of the size of the MEC is the following: As mentioned earlier, one can compare different DAGs to find the most consistent one with a given joint distribution over the variables. Here, the space of DAGs, or the space of MECs can be chosen as the search space. Chickering showed that when the size of the MECs are large, searching over the space of MECs is more efficient \cite{chickering2002optimal}. Hence, knowing the size of the MEC can help us decide the search space.

For the general problem of enumerating MECs, Steinsky proposed recursive enumeration formulas for the number of labelled essential graphs, in which the enumeration parameters are the number of vertices, chain components, and cliques \cite{steinsky2013enumeration}. This approach is not focused on a certain given MEC. For the  problem of finding the size of a given MEC, an existing solution is to use Markov chain methods \cite{he2013reversible,bernstein2017sampling}. According to this method, a Markov chain is constructed over the elements of the MEC whose properties ensure that the stationary distribution is uniform over all the elements. The rate of convergence and computational issues hinders practical application of Markov chain methods.
Recently, an exact solution for finding the size of a given MEC was proposed \cite{he2015counting}, in which the main idea was noting that sub-classes of the MEC with a fixed unique root variable partition the class. The authors show that there are five types of MECs whose sizes can be calculated with five formulas, and for any other MEC, it can be partitioned recursively into smaller sub-classes until the sizes of all subclasses can be calculated from the five formulas. An accelerated version of the method in \cite{he2015counting} is proposed in \cite{he2016formulas}, which is based on the concept of core graphs.

In this paper, we propose a new counting approach, in which the counting is performed on the clique tree representation of the graph. Compared to \cite{he2015counting}, our method provides us with a more systematic way for finding the orientation of the edges in a rooted sub-class, and enables us to use the memory in an efficient way in the implementation of the algorithm. Also, using clique tree representation enables us to divide the graph into smaller pieces, and perform counting on each piece separately (Theorem \ref{thm:edgecut}). We will show that for bounded degree graphs, the proposed solution is capable of computing the size of the MEC in polynomial time. The counting technique can be utilized for two main goals: {\bf (a) Uniform sampling}, and {\bf (b) Applying prior knowledge}.
As will be explained, for these goals, it is essential in our approach to have the ability of explicitly controlling the performed orientations in the given essential graph. Therefore, neither the aforementioned five formulas presented in (He, Jia, and Yu 2015), nor the accelerated technique in (He and Yu 2016) are suitable for these purposes.
 
{\bf (a) Uniform sampling:} In Section \ref{sec:samp}, we show that our counting technique can be used to uniformly sample from a given MEC. This can be utilized in many scenarios; followings are two examples.
1. \emph{Evaluating effect of an action in a causal system.} For instance, in a network of users, we may be interested in finding out  conveying a news to which user leads to the maximum spread of the news. This question could be answered if the causal structure is known, but we often do not have the exact causal DAG. Instead, we can resort to uniformly sampling from the corresponding MEC and evaluate the effect of the action on the samples.
2.  Given a DAG $D$ from a MEC, there are simple algorithms for generating the essential graph corresponding to the MEC \cite{meek1995causal,andersson1997characterization,chickering2002optimal}. However, for MECs with large size, it is not computationally feasible to form every DAG represented by the essential graph. In \cite{hoyer2009bayesian} the authors require to evaluate the score of all DAGs in MEC to find the one with the highest score. Evaluating scores on uniform samples provides an estimate for the maximum score.

{\bf (b) Applying prior knowledge:} In Section \ref{sec:prior}, we will further extend our counting and sampling techniques to the case that some prior knowledge regarding the orientation of a subset of the edges in the structure is available. Such prior knowledge may come from imposing a partial ordering of variables before conducting causal structure learning from both observational and interventional data \cite{scheines1998tetrad,hauser2012characterization,wang2017permutation}, or from certain model restrictions \cite{hoyer2012causal,rothenhausler2018causal,eigenmann2017structure}.
We will show that rooted sub-classes of the MEC allow us to understand whether a set of edges with unknown directions can be oriented to agree with the prior, and if so, how many DAGs in the MEC are consistent with such an orientation. 
We present the following two applications:
1. The authors in \cite{nandy2017estimating} proposed a method for estimating the causal effect of joint interventions from observational data. Their method requires extracting possible valid parent sets of the intervention nodes from the essential graph, along with their multiplicity. They proposed the joint-IDA method for this goal, which is exponential in the size of the chain component of the essential graph, and hence, becomes infeasible for large components. In Section \ref{sec:prior}, we show that counting with prior knowledge can solve this issue of extracting possible parent sets. 
2. In Section \ref{sec:int}, we provide another specific scenario in which our proposed methodology is useful. This application is concerned with finding the best set of variables to intervene on when we are restricted to a certain budget for the number of interventions \cite{ghassami2018budgeted}.

\section{Definitions and Problem Description}
\label{sec:desc}

For the definitions in this section, we mainly follow Andersson \cite{andersson1997characterization}. A graph $G$ is a pair $(V(G),E(G))$, where $V(G)$ is a finite set of vertices and $E(G)$, the set of edges, is a subset of $(V\times V)\setminus\{(a,a):a\in V\}$. An edge $(a,b)\in E$ whose opposite $(b,a)\in E$ is called an undirected edge and we write $a-b\in G$. An edge $(a,b)\in E$ whose opposite $(b,a)\not\in E$ is called a directed edge, and we write $a\rightarrow b\in E$. A graph is called a chain graph if it contains no \emph{partially} directed cycles. After removing all directed edges of a chain graph, the components of the remaining undirected graph are called the chain components of the chain graph. A v-structure of $G$ is a triple $(a,b,c)$, with induced subgraph $a\rightarrow b\leftarrow c$. Under Markov condition and faithfulness assumptions, a directed acyclic graph (DAG) represents the conditional independences of a distribution on variables corresponding to its vertices. Two DAGs are called Markov equivalent if they represent the same set of conditional independence relations. The following result due to \cite{verma1990equivalence} provides a graphical test for Markov equivalence:
\begin{lemma}\cite{verma1990equivalence}
\label{lem:verma}
Two DAGs are Markov equivalent if and only if they have the same skeleton and v-structures.
\end{lemma}
 We denote the Markov equivalence class (MEC) containing DAG $D$ by $[D]$. A MEC can be represented by a graph $G^*$, called essential graph, which is defined as $G^*=\cup(D:D\in [D])$. We denote the MEC corresponding to essential graph $G^*$ by $\MEC(G^*)$. Essential graphs are also referred to as completed partially directed acyclic graphs (CPDAGs) \cite{chickering2002optimal}, and maximally oriented graphs \cite{meek1995causal} in the literature. \cite{andersson1997characterization} proposed a graphical criterion for characterizing an essential graph. They showed that an essential graph is a chain graph, in which every chain component is chordal. As a corollary of Lemma \ref{lem:verma}, no DAG in a MEC can contain a v-structure in the subgraphs corresponding to a chain component.

We refer to the number of elements of a MEC by the size of the MEC, and we denote the size of $\MEC(G^*)$ by $\textit{Size}(G^*)$. Let $\{G_1,...,G_c\}$ denote the chain components of $G^*$. $\textit{Size}(G^*)$ can be calculated from the size of chain components using the following equation \cite{gillispie2002size,he2008active}:
\begin{equation}
\label{eq:prod}
\textit{Size}(G^*)=\prod_{i=1}^c\textit{Size}(G_i).
\end{equation}
This result suggests that the counting can be done separately in every chain component.
Therefore, without loss of generality, we can focus on the problem of finding the size of a MEC for which the essential graph is a chain component, i.e., an undirected connected chordal graph (UCCG), in which none of the members of the MEC are allowed to contain any v-structures.

In order to solve this problem, the authors of \cite{he2015counting} showed that there are five types of MECs whose sizes can be calculated with five formulas, and that for any other MEC, it can be partitioned recursively into smaller subclasses until the sizes of all subclasses can be calculated from the five formulas. 
We explain the partitioning method as it is relevant to this work as well.

Let $G$ be a UCCG and $D$ be a DAG in $\MEC(G)$. A vertex $v\in V(G)$ is the root in $D$ if its in-degree is zero.
\begin{definition}
	Let UCCG $G$ be the representer of a MEC. The $v$-rooted sub-class of the MEC is the set of all $v$-rooted DAGs in the MEC. This sub-class can be represented by the $v$-rooted essential graph $G^{(v)}=\cup(D:D\in \text{$v$-rooted sub-class})$.
\end{definition}
For instance, for UCCG $G$ in Figure \ref{fig:ex1}$(a)$, $G^{(v_1)}$ and $G^{(v_2)}$ are depicted in Figures \ref{fig:ex1}$(c)$ and \ref{fig:ex1}$(g)$, respectively.
\begin{lemma}
\label{lem:partition}
	\cite{he2015counting} Let $G$ be a UCCG. For any $v\in V(G)$, the $v$-rooted sub-class is not empty and the set of all $v$-rooted sub-classes partitions the set of all DAGs in the MEC.
\end{lemma}
From Lemma \ref{lem:partition} we have
\begin{equation}
\label{eq:sum}
\textit{Size}(G)=\sum_{v\in V(G)}\textit{Size}(G^{(v)}).
\end{equation}
Hence, using equations \eqref{eq:prod} and \eqref{eq:sum}, we have
\begin{equation}
\label{eq:size}
\textit{Size}(G^*)=\prod_{i=1}^c\sum_{v\in V(G_i)}\textit{Size}(G_i^{(v)}).
\end{equation}
He et al. showed that $G^{(v)}$ is a chain graph with chordal chain components, and introduced an algorithm for generating this essential graph. Therefore, using equation \eqref{eq:size}, $\textit{Size}(G^*)$ can be obtained recursively.
The authors of \cite{he2015counting} did not characterize the complexity, but reported that their experiments suggested when the number of vertices is small or the graph is sparse, the proposed approach is efficient.

\section{Calculating the size of a MEC}
\label{sec:count}

In this section, we present our proposed method for calculating the size of the MEC corresponding to a given UCCG. We first introduce some machinery required in our approach for representing chordal graphs via clique trees. The definitions and propositions are mostly sourced from \cite{blair1993introduction,vandenberghe2015chordal}.

For a given UCCG, $G$, let $\mathcal{K}_G=\{K_1,\cdots,K_m\}$ denote the set containing the maximal cliques of $G$, and let $T=(\mathcal{K}_G,E(T))$ be a tree on $\mathcal{K}_G$, referred to as a \emph{clique tree}.

\begin{definition}[\emph{Clique-intersection property}]
	A clique tree $T$ satisfies the clique-intersection property if for every pair of distinct cliques $K, K'\in\mathcal{K}_G$, the set $K\cap K'$ is contained in every clique on the path connecting $K$ and $K'$ in the tree.
\end{definition}

\begin{proposition}
\label{prop:blair1}
	\cite{blair1993introduction} A connected graph $G$ is chordal if and only if there exists a clique tree 
	for $G$, which satisfies the clique-intersection property.
\end{proposition}

\begin{definition}[\emph{Induced-subtree property}]
	A clique tree $T$ satisfies the induced-subtree property if for every vertex $v\in V(G)$, the set of cliques containing $v$ induces a subtree of $T$, denoted by $T_v$.
\end{definition}

\begin{proposition}
\label{prop:blair2}
	\cite{blair1993introduction} The clique-intersection and induced-subtree properties are equivalent.
\end{proposition}

Since we work with a UCCG, by Propositions \ref{prop:blair1} and \ref{prop:blair2}, there exists a clique tree on $\mathcal{K}_G$, which satisfies the clique-intersection and induced-subtree properties. Efficient algorithms for generating such a tree is presented in \cite{blair1993introduction,vandenberghe2015chordal}. In the sequel, whenever we refer to clique tree $T=(\mathcal{K}_G,E(T))$, we assume it satisfies the clique-intersection and induced-subtree properties.
 Note that such a tree is not necessarily unique, yet we fix one and perform all operations on it. For the given UCCG, similar to \cite{he2015counting}, we partition the corresponding MEC by its rooted sub-classes, and calculate the size of each sub-class separately. However, we use the clique tree representation of the UCCG for counting.
Our approach enables us to use the memory in the counting process to make the counting more efficient, and provides us with a systematic approach for finding the orientations.


For the $r$-rooted essential graph $G^{(r)}$, we arbitrarily choose one of the cliques containing $r$ as the root of the tree $T=(\mathcal{K}_G,E(T))$, and denote this rooted clique tree by $T^{(r)}$. Setting a vertex as the root in a tree determines the parent of all vertices. For clique $K$ in $T^{(r)}$, denote its parent clique by $\Pa(K)$  . Following \cite{vandenberghe2015chordal}, we can partition each non-root clique into a separator set $\Sep(K)=K\cap\Pa(K)$, and a residual set $\Res(K)=K\setminus\Sep(K)$. For the root clique, the convention is to define $\Sep(K)=\emptyset$, and $\Res(K)=K$. The induced-subtree property implies the following result:
\begin{proposition}
\label{prop:van}
\cite{vandenberghe2015chordal}	
Let $G$ be the given UCCG, and let $T^{(r)}$ be the rooted clique tree,\\
(i) The clique residuals partition $V(G)$,\\
(ii) For each $u\in V(G)$, denote the clique for which $u$ is in its residual set by $K_u$, and the induced subtree of cliques containing $u$ by $T_u$. $K_u$ is the root of $T_u$. The other vertices of $T_u$ are the cliques that contain $u$ in their separator set.\\
(iii) The clique separators $\Sep(K)$, where $K$ ranges over all non-root cliques, are the minimal vertex separators of $G$.
\end{proposition}
Note that sets $\Sep(K)$ and $\Res(K)$ depend on the choice of root. However, all clique trees have the same vertices (namely, maximal cliques of $G$) and the same clique separators (namely, the minimal vertex separators of $G$). Also, since there are at most $p$ cliques in a chordal graph with $p$ vertices, there are at most $p-1$ edges in a clique tree and, hence, at most $p-1$ minimal vertex separators. 

With a small deviation from the standard convention, in our approach, for the root clique of $T^{(r)}$, we define $\Sep(K)=\{r\}$, and $\Res(K)=K\setminus\{r\}$.
Recall from Proposition \ref{prop:van} that for each $u\in V(G)$, $K_u$ denotes the clique for which $u$ is in its residual set, and this clique is unique. We will need the following results for our orientation approach. All the proofs are provided in the appendix.
\begin{lemma}
\label{lem:main}
	If $u\rightarrow v\in G^{(r)}$, then $u\in\Sep(K_v)$.
\end{lemma}
\begin{corollary}
\label{cor:dir}
	If $u\rightarrow v\in G^{(r)}$, then $v\not\in\Sep(K_u)$.
\end{corollary}
Lemma \ref{lem:main} states that parents of any vertex $v$ are elements of $\Sep(K_v)$. But, not all elements of $\Sep(K_v)$ are necessarily parents of $v$. Our objective is to find a necessary and sufficient condition to determine the parents of a vertex.
\begin{lemma}
\label{lem:all}
	If $u\rightarrow v\in G^{(r)}$, then for every vertex  $w\in\Res(K_v)$, $u\rightarrow w\in G^{(r)}$.
\end{lemma}
Lemma \ref{lem:all} implies that for every clique $K$, any vertex $u\in\Sep(K)$ either has directed edges to all elements of $\Res(K)$ in $G^{(r)}$, or has no directed edges to any of the elements of $\Res(K)$. For clique $K$, we define the \emph{emission set}, $\Em(K)$, as the subset of $\Sep(K)$ which has directed edges to all elements of $\Res(K)$ in $G^{(r)}$. Lemmas \ref{lem:main} and \ref{lem:all} lead to the following corollary:
\begin{corollary}
\label{cor:main}
	$\Em(K_v)$ is the set of parents of $v$ in $G^{(r)}$.
\end{corollary}
This means that the necessary and sufficient condition for $u$ to be a parent of $v$ is $u\in\Em(K_v)$. Therefore, for any clique $K$, we need to characterize its emission set. We will show that the $\Em(K)$ is the subset of $\Sep(K)$ which satisfies the following emission condition:
\begin{definition}[Emission condition for a vetex]
\label{def:emissionV}
	We say vertex $v$ satisfies the emission condition in clique $K$ if $v\in\Sep(K)$, and $\Em(K_v)\not\subseteq\Sep(K)$.
\end{definition}
As a convention, we assume that the root variable $r$ satisfies the emission condition in all the cliques containing it.
\begin{theorem}
\label{thm:main1}
	$u\rightarrow v\in G^{(r)}$ if and only if $u$ satisfies the emission condition in clique $K_v$ in $T^{(r)}$.
\end{theorem}
Note that for a clique $K$, in order to find $\Em(K)$ using Definition \ref{def:emissionV}, we need to learn the emission set of some cliques on the higher levels of the tree. Hence, the emission sets must be identified on the tree from top to bottom.

\begin{figure}[t]
\begin{center}
\includegraphics[scale=0.55]{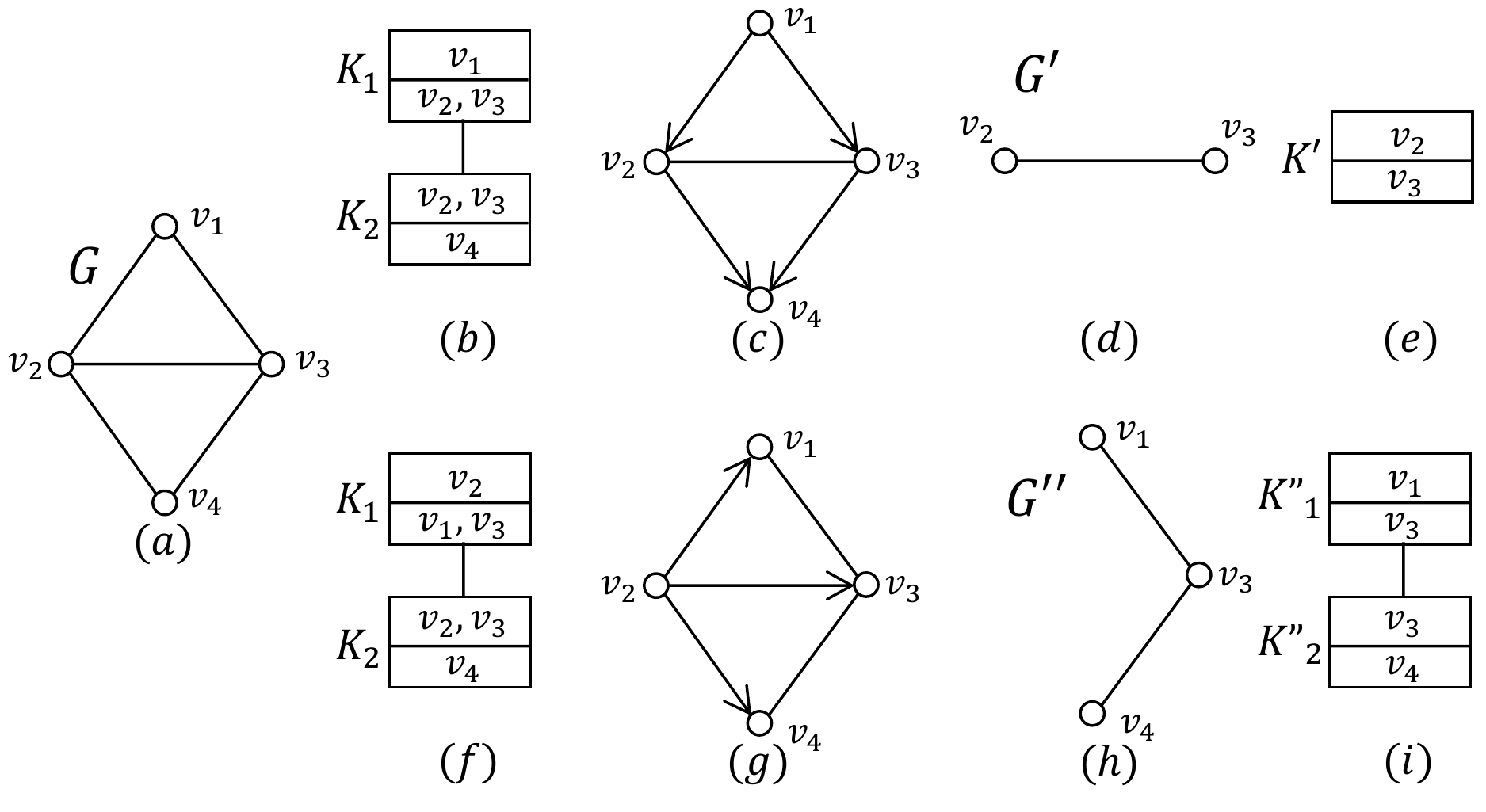}
\caption{Graphs related to Example \ref{ex:1}.}
\label{fig:ex1}
\end{center}
\end{figure}

After setting vertex $r$ as the root, Theorem \ref{thm:main1} allows us to find the orientation of directed edges in $G^{(r)}$ as follows. First, we form $T^{(r)}$. Then, for each $K\in\mathcal{K}_G$, starting from top to bottom of $T^{(r)}$, we identify all vertices which satisfy the emission condition to obtain $\Em(K)$.  
Finally, in each clique $K$, we orient the edges from all the variables in $\Em(K)$ towards all the variables in $\Res(K)$.
\begin{example}
\label{ex:1}
Assume the UCCG in Figure \ref{fig:ex1}$(a)$ is the given essential graph.
Setting vertex $v_1$ as the root of $G$ (by symmetry, $v_4$ is similar), the corresponding clique tree $T^{(v_1)}$ is shown in Figure \ref{fig:ex1}$(b)$, where in each clique, the first and the second rows represent the separator and the residual sets, respectively. In this clique tree, we obtain $\Em(K_1)=\{v_1\}$, and $\Em(K_2)=\{v_2,v_3\}$. Hence, the directed edges are $v_1\rightarrow v_2$, $v_1\rightarrow v_3$, $v_2\rightarrow v_4$, and $v_3\rightarrow v_4$. This results in $G^{(v_1)}$ in Figure \ref{fig:ex1}$(c)$, which is an essential graph with a single chain component $G'$, (Figure \ref{fig:ex1}$(d)$).
Setting vertex $v_2$ as the root (by symmetry, $v_3$ is similar), the corresponding clique tree $T'^{(v_2)}$ is shown in Figure \ref{fig:ex1}$(e)$. In this clique tree, $\Em(K')=\{v_2\}$ and hence, the directed edge is $v_2\rightarrow v_3$. This results in a directed graph, thus, $Size(G'^{(v_2)})=1$. Similarly, $Size(G'^{(v_3)})=1$. Therefore, from \eqref{eq:sum}, we have $Size(G^{(v_1)})=Size(G'^{(v_2)})+Size(G'^{(v_3)})=2$. Similarly, we have $Size(G^{(v_4)})=2$.

Setting vertex $v_2$ as the root of $G$ (by symmetry, $v_3$ is similar), the corresponding clique tree $T^{(v_2)}$ is shown in Figure \ref{fig:ex1}$(f)$. In this clique tree, $\Em(K_1)=\Em(K_2)=\{v_2\}$, and hence, the directed edges are $v_2\rightarrow v_1$, $v_2\rightarrow v_3$, and $v_2\rightarrow v_4$. $G^{(v_2)}$ shown in Figure \ref{fig:ex1}$(g)$ is the essential graph with a single chain component $G''$, shown in Figure \ref{fig:ex1}$(h)$.
Setting vertex $v_1$ as the root, the corresponding clique tree $T''^{(v_1)}$ is shown in Figure \ref{fig:ex1}$(i)$. In this clique tree $\Em(K''_1)=\{v_1\}$ and $\Em(K''_2)=\{v_3\}$. Hence, the directed edges are $v_1\rightarrow v_3$ and $v_3\rightarrow v_4$. The result is a directed graph, thus, $Size(G''^{(v_1)})=1$. Similarly, $Size(G''^{(v_3)})=1$ and $Size(G''^{(v_4)})=1$. Therefore, from \eqref{eq:sum}, we have $Size(G^{(v_2)})=Size(G''^{(v_1)})+Size(G''^{(v_3)})+Size(G''^{(v_4)})=3$. Similarly, we have $Size(G^{(v_3)})=3$.
 
Finally, using equation \eqref{eq:sum}, we obtain that $Size(G)=\sum_i Size(G^{(v_i)})=10$.
\end{example}



\begin{algorithm}[t]
\begin{algorithmic}[1]
\State {\bf Input:} Essential graph $G^*$, with chain components
\State \hspace{10mm} $\{G_1,\cdots G_c\}$.
\State {\bf Return:} $\prod_{i=1}^c\textsc{Size}(G_i)$\\ 
\hrulefill
\Function{Size}{$G$}
\State Construct a clique tree $T=(\mathcal{K}_G,E(T))$.
\If{$T\in$ Memory}
\State Load $[T,Size_{T}]$, {\bf Return:} $Size_{T}$
\Else
\For{$v\in V(G)$}
\State Set a clique $K\in T_v$ as the root to form $T^{(v)}$.
\State $Size_{T^{(v)}}=\textsc{RS}(T^{(v)}, K, \Sep(K))$.
\EndFor
\State Save $[T, \sum_{v\in V} Size_{T^{(v)}}]$, {\bf Return} $\sum_{v} Size_{T^{(v)}}$
\EndIf
\EndFunction
\caption{MEC Size Calculator}
\label{algorithm:size}
\end{algorithmic}
\end{algorithm}

\subsection{Algorithm}
\label{sec:algorithm}

In this subsection, we present an efficient approach for the counting process. 
In a rooted clique tree $T^{(r)}$, for any clique $K$, let $T^{(K)}$ be the maximal subtree of $T^{(r)}$ with $K$ as its root, and let $\Res(T^{(K)}):=\bigcup_{K'\in T^{(K)}}\Res(K')$.
Also, for vertex sets $S_1,S_2\subseteq V$, let $[S_1,S_2]$ be the set of edges with one end point in $S_1$ and the other end point in $S_2$.
\begin{lemma}
\label{lem:edgecut}
	For any clique $K$, $[\Sep(K),\Res(T^{(K)})]$ is an edge cut.
\end{lemma}
We need the following definition in our algorithm:
\begin{definition}[Emission condition for a clique]
\label{def:emissionC}
Clique $K$ satisfies the emission condition if $\Em(K)=\Sep(K)$.
\end{definition}

\begin{remark}
\label{rmk:em}
Theorem \ref{thm:main1} implies that clique $K$ satisfies the emission condition if and only if all elements in $\Sep(K)$ satisfy the emission condition in $K$.
\end{remark}
In the recursive approach, once the algorithm finds all the directed edges in $G^{(r)}$, it removes all the oriented edges for the next stage and restarts with the undirected components, i.e., the edges that it removes are the directed edges. Therefore, we require an edge cut in which all the edges are directed. This is satisfied by cliques with emission condition:
\begin{theorem}
\label{thm:edgecut}
		If clique $K$ satisfies the emission condition, $[\Sep(K),\Res(T^{(K)})]$ is a directed edge cut.
\end{theorem}
\begin{algorithm}[t]
\addtocounter{algorithm}{-1}
\renewcommand{\thealgorithm}{Function}
\floatname{algorithm}{}
\begin{algorithmic}[1]
\State {\bf Initiate:} $Size_T=1$
\State Orient from $\Sep(K_{root})$ to $\Res(K_{root})$ in $G$.
\State $\Ex=K_{root}$
\While{$\Ex\neq \emptyset$}
\For{$K\in \Ch(\Ex)$}
\State Form $\Em(K)$.
\If{$K$ satisfies emission condition}
\State $T=T\backslash T^{(K)}$
\State Remove $[\Sep(K), \Res(T^{(K)})]$ and components containing $\Res(T^{(K)})$ from $G$.
\If{$(T^{(K)},\Sep(K))\in$ Memory}
\State Load $[(T^{(K)}, \Sep(K)),Size_{T^{(K)}}]$
\State $Size_{T}=Size_{T}\times Size_{T^{(K)}}$
\Else
\State $Size_{T^{(K)}}=\textsc{RS}(T^{(K)},K,\Sep(K))$
\State Save $[(T^{(K)}, \Sep(K)),Size_{T^{(K)}}]$
\State $Size_{T}=Size_{T}\times Size_{T^{(K)}}$
\EndIf
\Else
\State Orient from $\Em(K)$ to $\Res(K)$ in $G$.
\EndIf
\EndFor
\State $\Ex=\Ch(\Ex)$
\EndWhile
\State {\bf Return:} $Size_T\times\prod_{G'\in\text{chain comp of $G$}}\textsc{Size}(G')$.
\caption{$\textsc{RS}(T,K_{root},\Sep(K_{root}))$}
\end{algorithmic}
\end{algorithm}

Therefore, by Theorem \ref{thm:edgecut}, if clique $K$ satisfies the emission condition, in the clique tree, we can learn the orientations in trees $T^{(r)}\setminus T^{(K)}$ and $T^{(K)}$ separately. This property helps us to perform  our counting process efficiently by utilizing the memory in the process. More specifically, if rooted clique trees $T^{(r_1)}$ and $T^{(r_2)}$ share a rooted subtree whose root clique satisfies the emission condition, it suffices to perform the counting in this subtree only once. Based on this observation, we propose the  counting approach whose pseudo-code is presented in Algorithm \ref{algorithm:size}.

The input to Algorithm \ref{algorithm:size} is an essential graph $G^*$, and it returns $\Size(G^*)$ by computing equation \eqref{eq:prod}, through calling function $\textsc{Size}(\cdot)$ for each chain component of $G^*$. In Function $\textsc{Size}(\cdot)$, first the clique tree corresponding to the input UCCG is constructed. If this tree has not yet appeared in the memory, for every vertex $v$ of the input UCCG, the  function forms $T^{(v)}$ and calculates the number of $v$-rooted DAGs in the MEC by calling the rooted-size function $\textsc{RS}(\cdot)$ (lines 10-13). Finally, it saves and returns the sum of the sizes.

Function $\textsc{RS}(\cdot)$ checks whether each clique $K$ of each level of the input rooted tree satisfies the emission condition (lines 7). If so, it removes the rooted subtree $T^{(K)}$ from the tree and the corresponding subgraph from $G$ (lines 8-10), and checks whether the size of $T^{(K)}$ with its current separator set is already in the memory. If it is, the function loads it as $\Size_{T^{(K)}}$ (lines 12); else, $\textsc{RS}(\cdot)$ calls itself on the rooted clique tree $T^{(K)}$ to obtain $Size_{T^{(K)}}$, and then saves $\Size_{T^{(K)}}$ in the memory (lines 15 and 16). If the clique $K$ does not satisfy the emission condition, it simply orients edges from $\Em(K)$ to $\Res(K)$ in $G$ (lines 19). Finally, in the resulting essential graph, it calls the function $\textsc{Size}(\cdot)$ for each chain component (lines 24). 

For bounded degree graphs, the proposed approach runs in polynomial time:

\begin{theorem}
\label{th:compSIZE}
	Let $p$ and $\Delta$ be the number of vertices and maximum degree of a graph $G$. The computational complexity of MEC size calculator on $G$ is in the order of $O(p^{\Delta+2})$.
\end{theorem}

\begin{remark}
From Definition \ref{def:emissionV}, it is clear that if $\Sep(K_v)\cap\Sep(K)=\emptyset$, then $v\in\Sep(K)$ satisfies the emission condition in clique $K$. We can use this property for locally orienting edges without finding emission set $\Em(K_v)$.	
\end{remark}

\subsection{Simulation Results}

\begin{table}[t]
\begin{center}
  \begin{tabular}{ |c | c c c c c c |}
    \hline
$r$ & $p$ & 20 & 30 & 40 & 50 & 60 \\ \hline
\multirow{3}{*}{$0.2$}& $T_1$ &  0.50 & 2.26 & 6.65 & 19.55  & 55.59  \\
 & $T_2$  & 0.27 & 2.61 & 20.68 & 219.98 & $>$3600\\ 
 & $T_2/T_1$ & 0.54 & 1.15 & 3.11 & 11.25 & $>$65    \\ \hline
 \multirow{3}{*}{$0.25$} & $T_1$ & 0.51 & 2.27 & 7.56 & 25.46 & 59.21 \\ 
 & $T_2$ &  0.40 & 8.77 & 101.84 & 1760.21 & $>$3600\\
  &  $T_2/T_1$ & 0.78 & 3.86 & 13.47 & 69.12 & $>$60 \\
    \hline
  \end{tabular}
  \end{center}
  \caption{Average run time (in seconds).}
\label{table:counting}
\end{table}

We generated $100$ random UCCGs of size $p=20,\cdots,60$ with $r\times {p \choose 2}$ as the number of edges based on the procedure proposed in \cite{he2015counting}, where parameter $r$ controls the graph density. We compared the proposed algorithm with the counting algorithm in \cite{he2015counting} in Table \ref{table:counting}. 
Note that, as we mentioned earlier, since the counting methods are utilized for the purpose of sampling and applying prior knowledge, the five formulas in \cite{he2015counting} are not implemented in either of the counting algorithms.
The parameters $T_1$ and $T_2$ denote the average run time (in seconds) of the proposed algorithm and the counting algorithm in \cite{he2015counting}, respectively. 
For dense graphs, our algorithm is at least 60 times faster.

\section{Uniform Sampling from a MEC}
\label{sec:samp}

In this section, we introduce a sampler for generating random DAGs from a MEC. The  sampler is based on the counting method presented in Section 3. The main idea is to choose a vertex as the root according to the portion of members of the MEC having that vertex as the root, i.e., in UCCG $G$, vertex $v$ should be picked as the root with probability $\Size(G^{(v)})/ \Size(G)$.

The pseudo-code of our uniform sampler is presented in Algorithm \ref{algorithm:unifsamp}, which uses functions \textsc{Size}$(\cdot)$ and \textsc{RS}$(\cdot)$ of Section \ref{sec:algorithm}. The input to the sampler is an essential graph $G^*$, with chain components $\mathcal{G}=\{G_1,\cdots G_c\}$. For each chain component $G\in \mathcal{G}$, we set $v\in V(G)$ as the root with probability $\textsc{RS}(T^{(v)}, K, \Sep(K))/\textsc{Size}(G)$, where $K\in T_v$, and then we orient the edges in $G^*$ as in Algorithm \ref{algorithm:size}. We remove $G$ and add the created chain components to $\mathcal{G}$, and repeat this procedure until all edges are oriented, i.e., $\mathcal{G}=\emptyset$.

\begin{example}
For the UCCG in Figure \ref{fig:ex1}$(a)$, as observed in Example \ref{ex:1}, $\Size(G^{(v_1)})=\Size(G^{(v_4)})=2$,  $\Size(G^{(v_2)})=\Size(G^{(v_3)})=3$, and $\Size(G)=10$. Therefore, we set vertices $v_1$, $v_2$, $v_3$, and $v_4$ as the root with probabilities $2/10$, $3/10$, $3/10$, and $2/10$, respectively. Suppose $v_2$ is chosen as the root. Then as seen in Example \ref{ex:1}, $\Size(G''^{(v_1)})=\Size(G''^{(v_3)})=\Size(G''^{(v_4)})=1$. Therefore, in $G''$, we set either of the vertices as the root with equal probability to obtain the final DAG.
\end{example}

\begin{theorem}
\label{thm:unif}
The sampler in Algorithm \ref{algorithm:unifsamp} is uniform. 
\end{theorem}
For bounded degree graphs, the proposed sampler is capable of producing uniform samples in polynomial time.
\begin{corollary}
\label{cor:sampcomp}
The computational complexity of the uniform sampler is in the order of $O(\Delta p^{\Delta+2})$.
\end{corollary}

\begin{algorithm}[t]
\begin{algorithmic}
\State {\bf Input:} Essential graph $G^*$, with chain components
\State \hspace{10mm} $\mathcal{G}=\{G_1,\cdots G_c\}$.
\While{$\mathcal{G}\neq \emptyset$}
\State Pick an element $G\in\mathcal{G}$, and update $\mathcal{G}=\mathcal{G}\setminus G$.
\State Run $\textsc{Rooted}(G)$.
\EndWhile
\State {\bf Return:} $G^*$\\ 
\hrulefill
\Function{Rooted}{$G$}
\State Construct a clique tree $T=(\mathcal{K}_G,E(T))$.
\State Set $v\in V(G)$ as the root with prob. $\frac{\textsc{RS}(T^{(v)}, K, \Sep(K))}{\textsc{Size}(G)}$.
\State For every clique $K$ in $T^{(v)}$, form $\Em(K)$.
\State Orient from $\Em(K)$ to $\Res(K)$ in $G^*$ and $G$.
\State $\mathcal{G}=\mathcal{G}\cup\{\text{chain components of }G\}$.
\EndFunction
\caption{Uniform Sampler}
\label{algorithm:unifsamp}
\end{algorithmic}
\end{algorithm}

\section{Counting and Sampling with Prior Knowledge}
\label{sec:prior}

Although in structure learning from observational data the orientation of some edges may remain unresolved, in many applications, an expert may have prior knowledge regarding the direction of some of the unresolved edges. In this section, we extend the counting and sampling methods to the case that such prior knowledge about the orientation of a subset of the edges is available. 
Specifically, we require that in the counting task, only DAGs which are consistent with the prior knowledge are counted, and in the sampling task, we force all the generated sample DAGs to be consistent with the prior knowledge. 
Note that the prior knowledge may not be necessarily realizable, that is, there may not exist a DAGs in the corresponding MEC with the required orientations. In this case, the counting should return zero, and the sampling should return an empty set.

\subsection{Counting with Prior Knowledge}

We present the available prior knowledge in the form of a hypothesis graph $H=(V(H),E(H))$.
Consider an essential graph $G^*$. For $G^*$, we call a hypothesis realizable if there is a member of the MEC with directed edges consistent with the hypothesis. In other words, a hypothesis is realizable if the rest of the edges in $G^*$ can be oriented without creating any v-structures or cycles. More formally:
\begin{definition}
For an essential graph $G^*$, a hypothesis graph $H=(V(H),E(H))$ is called realizable if there exists a DAG $D$ in $\MEC(G^*)$, for which $E(D)\subseteq E(H)$.
\end{definition}

\begin{example}
For essential graph $v_1-v_2-v_3-v_4$, the hypothesis graph $H: v_1\rightarrow v_2-v_3\leftarrow v_4$ is not realizable, as the edge $v_2-v_3$ cannot be oriented without forming a v-structure.
\end{example}

For essential graph $G^*$, let $\Size_H(G^*)$ denote the number of the elements of $\MEC(G^*)$, which are consistent with hypothesis $H$, i.e., $\Size_H(G^*)=|\{D:D\in \text{MEC}(G^*),E(D)\subseteq E(H)\}|$. Hypothesis $H$ is realizable if $\Size_H(G^*)\neq 0$.

As mentioned earlier, each chain component $G$ of a chain graph contains exactly one root variable. We utilize this property to check the realizability and calculate $\Size_H(G^*)$ for a hypothesis graph $H$.
Consider essential graph $G^*$ with chain components $\mathcal{G}=\{G_1,\cdots C_c\}$. Following the same line of reasoning as in equation \eqref{eq:prod}, we have
\begin{equation}
\label{eq:prodH}
\textit{Size}_H(G^*)=\prod_{i=1}^c\textit{Size}_H(G_i).
\end{equation}
Also, akin to equation \eqref{eq:sum}, for any $G\in\mathcal{G}$,
\begin{equation}
\label{eq:sumH}
\textit{Size}_H(G)=\sum_{v\in V(G)}\textit{Size}_H(G^{(v)}).
\end{equation}
Therefore, in order to extend the pseudo code to the case of prior knowledge, we modify functions $\textsc{Size}(\cdot)$ and $\textsc{RS}(\cdot)$ to get $H$ as an extra input.
In our proposed pseudo code, the orientation task is performed in lines 2 and 19 of Function $\textsc{RS}(\cdot)$. Let $S$ be the set of directed edges of form $(u,v)$, oriented in either line 2 or 19. In function $\textsc{RS}(\cdot)$, after each of lines 2 and 19, we check the following:

\begin{algorithmic}
\If{$S\not\subseteq E(H)$}
\State {\bf Return:} 0
\EndIf
\end{algorithmic}
This guarantees that, any DAG considered in the counting will be consistent with the hypothesis $H$.

\begin{figure}[t]
\begin{center}
\includegraphics[scale=0.6]{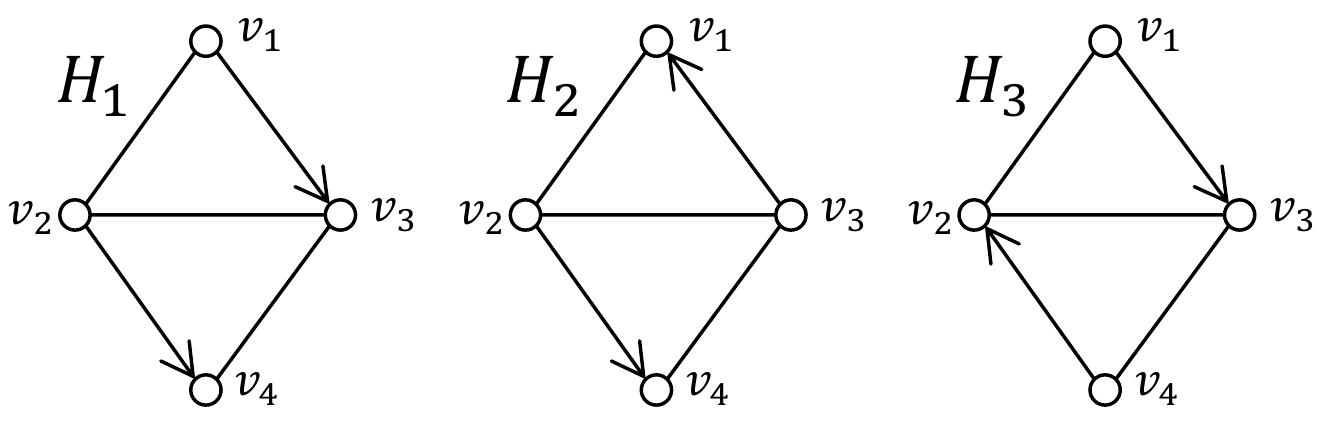}
\caption{Graphs related to Example \ref{ex:4}.}
\label{fig:ex4}
\end{center}
\end{figure}

\begin{example}
\label{ex:4}
Consider the three hypothesis graphs in Figure \ref{fig:ex4} for the essential graph in Figure \ref{fig:ex1}$(a)$. For hypothesis $H_1$, $\Size_{H_1}(G^{(v_1)})=2$, $\Size_{H_1}(G^{(v_2)})=1$, and $\Size_{H_1}(G^{(v_3)})=\Size_{H_1}(G^{(v_4)})=0$.
Therefore, we have three DAGs consistent with hypothesis $H_1$, i.e., $\Size_{H_1}(G)=3$. For hypothesis $H_2$, $\Size_{H_2}(G^{(v_2)})=\Size_{H_2}(G^{(v_3)})=2$, and $\Size_{H_2}(G^{(v_1)})=\Size_{H_2}(G^{(v_4)})=0$, 
Therefore, four DAGs are consistent with hypothesis $H_2$, i.e., $\Size_{H_2}(G)=4$. Hypothesis $H_3$ is not realizable.
\end{example}

One noteworthy application of checking the realizability of a hypothesis is in the context of estimating the causal effect of interventions from observational data \cite{maathuis2009estimating,nandy2017estimating}.
This could be used for instance, to predict the effect of gene knockouts on other genes or some phenotype of interest, based on observational gene expression profiles. The authors of \cite{maathuis2009estimating,nandy2017estimating} proposed a method called (joint-)IDA for estimating the average causal effect, which as a main step requires extracting possible valid parent sets of the intervention nodes from the essential graph, with the multiplicity information of the sets.
To this end, a semi-local method was proposed in \cite{nandy2017estimating}, which is exponential in the size of the chain component of the essential graph. This renders the approach infeasible for large components. Using our proposed method to address this problem, we can fix a configuration for the parents of the intervention target, and count the number of consistent DAGs.


\subsection{Sampling with Prior Knowledge}

Suppose an experimenter is interested in generating sample DAGs from a MEC. However, due to her prior knowledge, she requires the generated samples to be consistent with a given set of orientations for a subset of the edges.
In this subsection, we modify our uniform sampler to apply to this scenario. We define the problem statement formally as follows. 
Given an essential graph $G^*$ and a hypothesis graph $H$ for $G^*$, we are interested in generating samples from the $\MEC(G^*)$ such that each sample is consistent with hypothesis $H$. Additionally, we require the distribution of the samples to be uniform conditioned on being consistent. That is, for each sample DAG $D\in\MEC(G^*)$, 
\[
P(D)=
\begin{cases}
\displaystyle\frac{1}{\Size_H(G^*)},&\text{if }E(D)\subseteq E(H),\\
\displaystyle0,\hspace{1mm}&\text{otherwise}.	
\end{cases}
\]


Equations \eqref{eq:prodH} and \eqref{eq:sumH} imply that we can use a method similar to the case of the uniform sampler. That is, we choose a vertex as the root according to the ratio of the DAGs $D\in\MEC(G^*)$ which are consistent with $H$ and have the chosen vertex as the root, to the total number of consistent DAGs. More precisely, in UCCG $G$, vertex $v$ should be picked as the root with probability $\Size_H(G^{(v)})/\Size_H(G)$. 
In fact, 
the uniform sampler could be viewed as a special case of sampler with prior knowledge with $H=G^*$. Hence, the results related to the uniform sampler extend naturally.

\begin{example}
Consider hypothesis graph $H_1$ in Figure \ref{fig:ex4} for the essential graph $G$ in Figure \ref{fig:ex1}$(a)$. As observed in Example \ref{ex:4}, we have $\Size_{H_1}(G^{(v_1)})=2$, $\Size_{H_1}(G^{(v_2)})=1$, $\Size_{H_1}(G^{(v_3)})=\Size_{H_1}(G^{(v_4)})=0$, and $\Size_{H_1}(G)=3$. Therefore, we set vertices $v_1$, $v_2$, $v_3$, and $v_4$ as the root with probabilities $2/3$, $1/3$, $0$, and $0$, respectively. 
\end{example}

\section{Application to Intervention Design}
\label{sec:int}


In this section, we demonstrate that the proposed method for calculating the size of MEC with prior knowledge can be utilized to design an optimal intervention target in experimental causal structure learning. We will use the setup in \cite{ghassami2018budgeted} which is as follows:
Let $G^*$ be the given essential graph, 
obtained from an initial stage of observational structure learning, 
and let $k$ be our intervention budget, i.e., the number of interventions we are allowed to perform. Each intervention is on only a single variable and the interventions are designed passively, i.e., the result of one intervention is not used for the design of the subsequent interventions. Let $\mathcal{I}$ denote the \emph{intervention target set}, which is the set of vertices that we intend to intervene on. Note that since the experiments are designed passively, this is not an ordered set. 
Intervening on a vertex $v$ resolves the orientation of all edges intersecting with $v$ \cite{eberhardt2005number}, and then we can run Meek rules to learn the maximal PDAG \cite{perkovic2017interpreting}. Let $R(\mathcal{I},D)$ be the number of edges that their orientation is resolved had the ground truth underlying DAG been $D$, and let $\mathcal{R}(\mathcal{I})$ be the average of $R(\mathcal{I},D)$ over the elements of the MEC, that is
\begin{equation}
\label{eq:R}
	\mathcal{R}(\mathcal{I})=\frac{1}{\Size(G^*)}\sum_{D\in\MEC(G^*)}R(\mathcal{I},D).
\end{equation}
The problem of interest is finding the set $\mathcal{I}\subseteq V(G^*)$ with $|\mathcal{I}|=k$, that maximizes $\mathcal{R}(\cdot)$.

In \cite{ghassami2018budgeted}, it was proved that $\mathcal{R}(\cdot)$ is a sub-modular function and hence, a greedy algorithm recovers an approximation to the optimum solution. Still, calculating $\mathcal{R}(\mathcal{I})$ for a given $\mathcal{I}$ remains as a challenge.
For an intervention target candidate, in order to calculate $\mathcal{R}(\mathcal{I})$, conceptually, we can list all DAGs in the MEC and then calculate the average according to \eqref{eq:R}. However, for large graphs, listing all DAGs in the MEC is computationally intensive.
Note that the initial information provided by an intervention is the orientation of the edges intersecting with the intervention target. Hence, we propose to consider this information as the prior knowledge and apply the method in Section \ref{sec:prior}. 

\begin{figure}[t]
\begin{center}
\includegraphics[height=45mm,width=90mm]{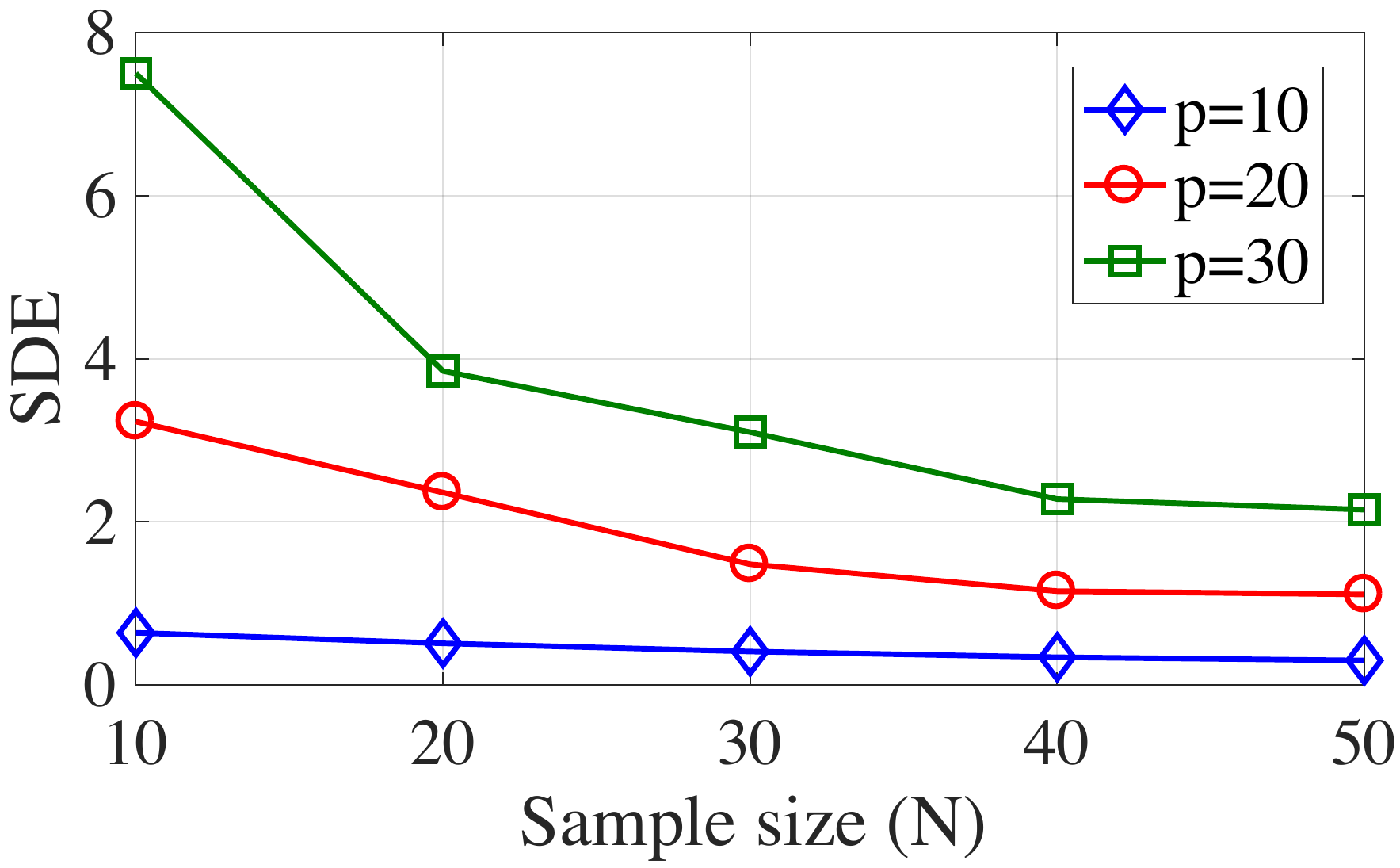}
\caption{SDE versus the sample size.}
\label{fig:sim}
\end{center}
\end{figure}

Let $\mathcal{H}$ be the set of hypothesis graphs, in which each element $H$ has a distinct configuration for the edges intersecting with the intervention target. If the maximum degree of the graph is $\Delta$, cardinality of $\mathcal{H}$ is at most $2^{k\Delta}$, and hence, it does not grow with $p$.
For a given hypothesis graph $H$, let $G^*_H=\{D: D\in MEC(G^*), E(D)\subseteq E(H)\}$ denote the set of members of the MEC, which are consistent with hypothesis $H$. Using the set $\mathcal{H}$, we can break \eqref{eq:R} into two sums as follows:
\begin{equation}
\label{eq:exact}	
\begin{aligned}	
\mathcal{R}(\mathcal{I})
	&=\frac{1}{\Size(G^*)}\sum_{D\in\MEC(G^*)}R(\mathcal{I},D)\\
	&=\frac{1}{\Size(G^*)}\sum_{H\in\mathcal{H}}\sum_{D\in G^*_H}R(\mathcal{I},D)\\
	&=\sum_{H\in\mathcal{H}}\frac{\Size_H(G^*)}{\Size(G^*)}R(\mathcal{I},D).
\end{aligned}
\end{equation}

Therefore, we only need to calculate at most $2^{k\Delta}$ values instead of considering all elements of the MEC, which reduces the complexity from super-exponential to constant in $p$.

\subsection{Simulation Results}
An alternative approach to calculating $\mathcal{R}(\mathcal{I})$ is to estimate its value by evaluating uniform samples. We generated 100 random UCCGs of size $p=10,20,30$ with $r\times{p \choose 2}$  edges, where $r=0.2$. In each graph, we selected two variables randomly to intervene on. We obtained the exact $\mathcal{R}(\mathcal{I})$ using equation \eqref{eq:exact}. Furthermore, for a given sample size $N$, we estimated $\mathcal{R}(\mathcal{I})$ from the aforementioned Monte-Carlo approach using our proposed uniform sampler and obtained empirical standard deviation of error (SDE) over all graphs with the same size, defined as $SD(|\mathcal{R}(\mathcal{I})-\hat{\mathcal{R}}(\mathcal{I})|)$. Figure \ref{fig:sim} depicts SDE versus the number of samples. As can be seen, SDE becomes fairly low for sample sizes greater than $40$.

\section{Conclusion}
We proposed a new technique for calculating the size of a MEC, which is based on the clique tree representation of chordal graphs. We demonstrated that this technique can be utilized for uniform sampling from a MEC, which provides a stochastic way to enumerate DAGs in the class, which can be used for estimating the optimum DAG, most suitable for a certain desired property.
We also extended our counting and sampling method to the case where prior knowledge about the structure is available, which can be utilized in applications such as causal intervention design and estimating the causal effect of joint interventions.

\begin{appendices}

\section{Proof of Lemma \ref{lem:main}}



Let $d_G(v,u)$ denote the distance between vertices $v$ and $u$ in $G$. The following result from \cite{bernstein2017sampling} is used in our proof. In the proof $r$ always denotes the root vertex.

\begin{lemma}
\label{lem:bern}
\cite{bernstein2017sampling}
In an acyclic and v-structure-free orientation of a UCCG $G$ the root variable $r$ determines the orientation of all edges $u-w$, for which $d_G(r,u)\neq d_G(r,w)$.
\end{lemma}

We partition vertices based on their distance from the root variable, and call each part a \emph{level}.  Note that clearly, there is no edge from level $l_i$ to $l_j$, for $j>i+1$, otherwise, the vertex in level $l_j$ should be moved to level $l_{i+1}$. Based on Lemma \ref{lem:bern}, the direction of edges in between the levels (mid-level edges) will be determined to be away from the root. After determining the mid-level edges, we should check for the direction of in-level edges as well. We will show that the statement of Lemma \ref{lem:main} holds for both mid-level and in-level edges.
\begin{itemize}
\item $u\rightarrow v\in G^{(r)}$ is a mid-level edge:

Proof by induction:

{\bf Induction base:} We need to show that for any vertex $w\in l_1$, $r\in\Sep(K_w)$. Let $K_{root}$ be the root clique. By definition, $r\in K_{root}$. For any vertex $w\in l_1$, by definition of $l_1$, $w$ is adjacent to the root, and hence, there exists a clique $K\in T_w$ such that $r\in K$. Therefore, by the clique-intersection property, $r$ is contained in every clique on the path connecting $K_{root}$ and $K$. Specifically, $r$ should be contained in $K_w$, as $K_w$ is the root of the subtree $T_w$, i.e., $r\in K_w$, otherwise, we will have a cycle in the tree. Noting that by our convention, $r$ only appears in separator sets, concludes that $r\in\Sep(K_w)$. Therefore, the base of the induction is clear.

{\bf Induction hypothesis:} As the induction hypothesis, we assume that for any variable $u'\in l_{i-1}$ and $v'\in l_i$, such that $u'\rightarrow v'\in G^{(r)}$, we have $u'\in\Sep(K_{v'})$.

{\bf Induction step:} Assume that $u\in l_{i}$ and $v\in l_{i+1}$, such that $u\rightarrow v\in G^{(r)}$. We need to show that $u\in\Sep(K_v)$.
\begin{claim}
\label{claim:1}
	$v\not\in K_u$.
\end{claim}
\begin{proof}
	Variable $v$ is non-adjacent to any variable in $l_j$, $j\le i-1$. Therefore, for all $w\in l_j$, $T_w\cap T_v=\emptyset$. On the other hand, $u$ should have a parent in $l_{i-1}$. Therefore, there exists $w\in l_{i-1}$, such that $T_w\cap T_u\neq\emptyset$. By induction hypothesis, $w\in\Sep(K_u)$. Therefore, since $T_w\cap T_v=\emptyset$, we have $v\not\in K_u$.
\end{proof}

Since $u$ and $v$ are adjacent, there exists a clique $K\in T_v$, such that $u\in K$. 
\begin{claim}
\label{claim:2}
	If $\{u,v\}\subseteq K$ and $v\not\in K_u$, then $u\in K_v$.
\end{claim}
\begin{proof}
By the induced-subtree property, since $K_v$ is the root of subtree $T_v$, the path from $K_{root}$ to $K$, passes through $K_v$. Similarly, the path from $K_{root}$ to $K$, passes through $K_u$. If $u\not\in K_v$, then the two aforementioned paths should be distinct, which results in a cycle in the tree, which is a contradiction.
\end{proof}

By Claim \ref{claim:2}, $u\in K_v$. Among all cliques that contain $u$, by Proposition \ref{prop:van}, $u$ is only in the residual set of $K_u$, and by Claim \ref{claim:1}, $v\not\in K_u$; therefore $K_u\neq K_v$. Therefore, $u\in \Sep(K_v)$.

\item $u\rightarrow v\in G^{(r)}$ is an in-level edge: Using the proof of Theorem 6 in \cite{he2008active}, if $u\rightarrow v\in G^{(r)}$, it should have been directed according to one of two possible rules: There exists vertex $w$ such that $G$ induced on $\{w,u,v\}$ is either (1) $w\rightarrow u-v$, or (2) $u\rightarrow w\rightarrow v$ and $u- v$. We show that the later case is not possible. This is similar to a claim in the proof of Theorem 8 in \cite{he2015counting}.\\ Proof by contradiction: Suppose $u-v$ is the first edge in level $l_i$ oriented by rule (2), that is, the other previously oriented edges are oriented via rule (1). Therefore, for the direction of edge $u\rightarrow w$, there should be $w_0$ in $l_{i-1}$ or $l_{i}$, such that $w_0\rightarrow u\in G$, and $w_0$ not adjacent to $w$. $w_0$ should also be adjacent to $v$, otherwise, we would learn $u-v$ from rule (1), not rule (2). Then to avoid cycle $\{w_0,u,w,v,w_0\}$, $w_0-v$ should be oriented as $w_0\rightarrow v$. But this directed edge will make a v-structure with $w\rightarrow v$, which is a contradiction. Therefore, we only need rule (1) for orientations. 

In order to orient edges using rule (1), we can apply the mid-level orienting method recursively. That is, we can consider the subgraph induced on a level and any vertex from the previous level as the root, and orient the mid-level edges for the new root.
\begin{claim}
\label{claim:3}
	Using rule (1) recursively is equivalent to applying the mid-level orienting method recursively.
\end{claim}
This claim is clear because if there exists induced subgraph $w\rightarrow u-v$ in $G$ on $\{w,u,v\}$, since in applying the mid-level orienting method recursively every vertex becomes root once, at the time that $w$ becomes root, the edge $u-v$ will fall in mid-levels. Therefore, it will be oriented away from $w$, i.e., it will be oriented as $u\rightarrow v$. Also, clearly we are not orienting any extra edges, as Lemma \ref{lem:bern}, is merely based on rule (1).


Therefore, by the previous part of the proof, the statement of Lemma \ref{lem:main} holds for this oriented edges as well. Applying this reasoning recursively concludes the desired result for all in-level edges.
\end{itemize}

\section{Proof of Corollary \ref{cor:dir}}
By Proposition \ref{prop:van}, for any vertex $w$, $K_w$ is the root of $T_w$. That is, among the cliques containing $w$, $K_w$ is located in the highest level (in terms of the distance from the root if the tree).

\emph{Proof by contradiction.} By Lemma \ref{lem:main}, $u\in\Sep(K_v)$ in $T^{(r)}$. That is, $K_v$ is in a strictly lower level than $K_u$ in the clique tree. Now if $v\in\Sep(K_u)$, $K_v$ should be in a strictly higher level than $K_u$ in the clique tree, which is a contradiction.

\section{Proof of Lemma \ref{lem:all}}
\emph{Proof by contradiction.} Suppose there exists $w\in \Res(K_v)$, such that $u\rightarrow w\not\in G^{(r)}$. As mentioned in Section 2, \cite{he2015counting} showed that $G^{(r)}$ is a chain graph with chordal chain components. Therefore, by the definition of chain graph, there should not be a partially directed cycle in this graph. Therefore, in order to prevent a partially directed cycle on $\{u,v,w,u\}$, we should have $w\rightarrow v\in G^{(r)}$. Therefore, by Lemma \ref{lem:main}, $w\in\Sep(K_v)$, and by Proposition \ref{prop:van}, $K_v$ is unique. Therefore, this is in contradiction to the assumption of the lemma.

\section{Proof of Corollary \ref{cor:main}}
By definition of the emissio set, $\Em(K_v)\subseteq\Pa(v)$. To prove the opposite direction, suppose there exists vertex $w$, such that $w\rightarrow v\in G^{(r)}$, but $w\not\in\Em(K_v)$. Then by Lemma \ref{lem:main},  $w\in\Sep(K_v)$, and by Proposition \ref{prop:van}, $K_v$ is unique. Now, since $w\rightarrow v\in G^{(r)}$ and $w\in\Sep(K_v)$, by Lemma \ref{lem:all}, $w\in\Em(K_v)$. This implies that $\Pa(v)\subseteq\Em(K_v)$. Therefore, $\Em(K_v)=\Pa(v)$.

\section{Proof of Theorem \ref{thm:main1}}

For the \textit{only if} side, we need to show that (1) $u\in\Sep(K_v)$, and (2) $\Em(K_u)\not\subseteq\Sep(K_v)$. (1) is obtained from Lemma \ref{lem:main}. We prove (2) by contradiction: Suppose $\Em(K_u)\subseteq\Sep(K_v)$. Then by Corollary \ref{cor:main}, $\Pa(u)\subseteq\Sep(K_v)$. Since $u\rightarrow v\in G^{(r)}$, as seen in the proof of Lemma \ref{lem:main}, this edge should have been directed from the induced subgraph $w\rightarrow u-v$, for some vertex $w$. But since $\Pa(u)\subseteq\Sep(K_v)$, every such vertex is in $K_v$, and hence, is adjacent to $v$. Therefore, the induced subgraph $w\rightarrow u-v$ cannot exist.

For the \textit{if} side, we first note that since $u,v\in K_v$, they are adjacent. By the assumption, 
\begin{align*}
&\Em(K_u)\not\subseteq\Sep(K_v)\\
\Rightarrow &\Pa(u)\not\subseteq\Pa(v)\\
\Rightarrow &\exists w \text{ such that}
\begin{cases}
w\rightarrow u\in G^{(r)},\\
w\rightarrow v\not\in G^{(r)}.
\end{cases}
\end{align*}
If $v\rightarrow w\in G^{(r)}$, or $v-w\in G^{(r)}$, in order to avoid a partially directed cycle on $\{w,u,v,w\}$, we should have $v\rightarrow u\in G^{(r)}$. Therefore, by Corollary \ref{cor:dir}, $u\not\in\Sep(K_v)$, and by Proposition \ref{prop:van}, $K_v$ is unique. Therefore, this is in contradiction to the assumption.\\
Therefore, $w$ and $v$ are not adjacent, and hence, we have the induced subgraph $w\rightarrow u-v$. This implies that in order to avoid v-structure, $u-v$ should be oriented as $u\rightarrow v$, which is the desired result.

\section{Proof of Lemma \ref{lem:edgecut}}

By Proposition \ref{prop:van}, we can partition the vertices into three sets: $\Res(T^{(K)})$, $\Sep(K)$, and $V(G)\setminus(\Sep(K)\cup\Res(T^{(K)}))$. 
By the clique-intersection property, any vertex in $T^{(r)}\setminus T^{(K)}$ which is not in $K$, will not be contained in $T^{(K)}$, and hence, will not be adjacent with any vertices in $\Res(T^{(K)})$. 
We note that vertices in $T^{(r)}\setminus T^{(K)}$ can only appear in the separator set of $K$.
Therefore, the set of vertices in $T^{(r)}\setminus T^{(K)}$ which are not in $K$ is $V(G)\setminus(\Sep(K)\cup\Res(T^{(K)}))$.

Therefore, vertices in $\Res(T^{(K)})$ are not adjacent with any vertex in $V(G)\setminus(\Sep(K)\cup\Res(T^{(K)}))$. 
Let $S:=\Res(T^{(K)})$, and $\bar{S}=V(G)\setminus S$.
Therefore, $[\Sep(K),\Res(T^{(K)})]=[S,\bar{S}]$, which implies that $[\Sep(K),\Res(T^{(K)})]$ is an edge cut.


\section{Proof of Theorem \ref{thm:edgecut}}

Due to Lemma \ref{lem:edgecut}, we only need to show that all the edges in $[\Sep(K),\Res(T^{(K)})]$ are directed.

By the definition of emission condition for cliques, $[\Sep(K),\Res(K)]$ is directed. For any variable $u\in\Sep(K)$, and $v\in\Res(T^{(K)})\setminus\Res(K)$, by Remark \ref{rmk:em}, $u$ satisfies the emission condition in $K$. Hence, if $u\in K_v$ by the clique-intersection property, and using Proposition \ref{prop:van}, $u$ satisfies the emission condition in clique $K_v$ as well. Therefore, by Theorem \ref{thm:main1}, $u\rightarrow v\in G$.

\section{Proof of Theorem \ref{th:compSIZE}}

In $\textsc{Size}$ function, for a given graph $G$, we first construct a clique tree $T$ in $O(p^2)$ by a modified version of maximum cardinality search \cite{blair1993introduction}. We assume that in the worst-case scenario, the memory condition in line $7$ of Algorithm is not satisfied in any recursive call. Thus, we set each vertex $v\in V(G)$ as the root and call $\textsc{RS}$ function to compute the number of DAGs in $G^{(v)}$. Moreover, in $\textsc{RS}$ function, in the worst-case scenario, we assume that the emission condition is not satisfied in any recursive call. Thus, the while loop in lines 4-24 orients all directed edges in $G^{(v)}$ by forming emission sets for each clique $K$ in the clique tree. In particular, in order to from $\Em(K)$, for any vertex $v$ in $\Sep(K)$, we check whether vertex $v$ satisfies emission condition which we can do it in linear time. Since the size of any clique $K$ is at most $p$, the computational complexity of obtaining $\Em(K)$ would be in the order of $O(p^2)$. Moreover, each clique $K$ is considered only once in the while loop. Thus, the directed edges in $G^{(v)}$ are recovered in $O(p^3)$ since we have at most $p$ cliques in the clique tree.

 Now, we show that the degree of each vertex $w$ in any chain component of $G^{(v)}$ decreases at least by one after removing directed edges. To do so, we prove that there exists a directed edge in $G^{(v)}$ that goes to vertex $w$. By contradiction, suppose that there is no such directed edge. Furthermore, assume that vertex $w$ is in a chain component $G'$. Consider the shortest path from $v$ to $w$ in $G^{(v)}$. This path must pass through one of neighbors of $w$ in $G'$ such as $u$. Since the distance from $v$ to $u$ is less than $v$ to $w$, $u-w$ should be oriented as $u\rightarrow w$ (see Lemma \ref{lem:bern}). But it is in contradiction with the fact that $u$ and $w$ are in the same chain component. Therefore, the degree of each vertex $w$ in any chain component of $G^{(v)}$ decreases at least by one after removing directed edges in $G^{(v)}$.

Let $t(\Delta)$ be the computational complexity of running $\textsc{Size}$ function on a graph with maximum degree $\Delta$. Based on what we proved above, we have
\begin{equation*}
t(\Delta)\leq p t(\Delta-1)+cp^3,
\end{equation*}  
where $c$ is a constant. The above inequality holds true since we have at most $p$ chain component in $G^{(v)}$ where the maximum degree in each of them is at most $\Delta-1$. From this inequality, it can be easily shown that $t(\Delta)$ is in the order of $O(p^{\Delta+1})$. Since we may have at most $p$ chain components in essential graph $G^*$, the computational complexity of MEC size calculator is in the order of $O(p^{\Delta+2})$.



\section{Proof of Theorem \ref{thm:unif}}

The objective is to show that for the input essential graph $G^*$, any DAG $D$ in the MEC represented by $G^*$ is generated with probability $1/Size(G^*)$.

\emph{Proof by induction:} The function $\textsc{Size}(\cdot)$ finds the size of a component recursively, i.e., after setting a vertex $v$ as the root, and finding the orientations in $G^{(v)}$, it calls itself to obtain the size of the chain components of $G^{(v)}$. We induct on the maximum number of recursive calls required for complete orienting.\\
{\bf Induction base:}
For the base of the induction, we consider an essential graph with no required recursive call: Consider essential graph $G^*$ with chain component set $\mathcal{G}$, for which, for all $G\in\mathcal{G}$, for all $v\in V(G)$, $Size(G^{(v)})=1$ (as an example, consider the case that $G$ is a tree). Consider $D$ in the MEC represented by $G^*$, and assume vertex $v_G$ is required to be set as the root in chain component $G\in\mathcal{G}$ for $D$ to be obtained. We have
\begin{align*}
P(D) &= \prod_{G\in\mathcal{G}}P(v_G \text{ picked})
=\prod_{G\in\mathcal{G}}\frac{Size(G^{(v)})}{Size(G)}\\
&=\prod_{G\in\mathcal{G}}\frac{1}{Size(G)}
=\frac{1}{\prod_{G\in\mathcal{G}}Size(G)}\\
&=\frac{1}{Size(G^*)},
\end{align*}
where, the last equality follows from equation \eqref{eq:prod}.\\
{\bf Induction hypothesis:} 
For an essential graph $G^*$ with maximum required recursions of $l-1$, any DAG $D$ in the MEC represented by $G^*$ is generated with probability $1/Size(G^*)$.\\
{\bf Induction step:} We need to show that for an essential graph $G^*$ with maximum required recursions of $l$, any DAG $D$ in the MEC represented by $G^*$ is generated with probability $1/Size(G^*)$.
Assume vertex $v_G$ is required to be set as the root in chain component $G\in\mathcal{G}$, and $V_{G^{(v)}}$ is the set of vertices required to be set as root in the next recursions in obtained chain components in $G^{(v)}$ for $D$ to be obtained. We have
\begin{align*}
P(D) &= \prod_{G\in\mathcal{G}}P(v_G \text{ picked})P(V_{G^{(v)}} \text{ picked})\\
&= \prod_{G\in\mathcal{G}}\frac{Size(G^{(v)})}{Size(G)}P(V_{G^{(v)}} \text{ picked}).
\end{align*}
By the induction hypothesis, $$P(V_{G^{(v)}} \text{ picked})=1/Size(G^{v}).$$ Therefore,
\begin{align*}
P(D) &= \prod_{G\in\mathcal{G}}\frac{Size(G^{(v)})}{Size(G)}\frac{1}{Size(G^{(v)})}\\
&=\frac{1}{\prod_{G\in\mathcal{G}}Size(G)}\\
&=\frac{1}{Size(G^*)},
\end{align*}
where, the last equality follows from equation \eqref{eq:prod}.

\section{Proof of Corollary \ref{cor:sampcomp}}
 In $\textsc{ROOTED}$ function, for any chain component $G$ in $\mathcal{G}$, we construct a clique tree in $O(p^2)$ \cite{blair1993introduction} and obtain probabilities of the from $\frac{\textsc{RS}(T^{(v)}, K, \Sep(K))}{\textsc{Size}(G)}$ by running $\textsc{Size}(G)$, which its computational complexity is $O(p^{\Delta+1})$ (see Theorem \ref{th:compSIZE}). After selecting one of the vertices in $G$ as the root, say $v$, we recover all directed edges in $G^{(v)}$ in $O(p^3)$ and obtain chain components of $G^{(v)}$. Similar to the proof of Theorem \ref{th:compSIZE}, let $t(\Delta)$ be the running time of the algorithm on a chain component in $\mathcal{G}$ with maximum degree of $\Delta$. Then, based on what we argued above, we have
\begin{equation*}
t(\Delta)\leq pt(\Delta-1)+c(p^{\Delta+1}+p^3),
\end{equation*}
where $c$ is a constant. It can be shown that $t(\Delta)$ is in the order of $O(\Delta p^{\Delta+1})$. Since we may have at most $p$ chain components in $\mathcal{G}$, the computational complexity of uniform sampler would be in the order of $O(\Delta p^{\Delta+2})$.

\end{appendices}


\newpage
\bibliographystyle{alpha}
\bibliography{Refs.bib}


\end{document}